\documentclass[11pt]{article}

\usepackage[english]{babel}

\usepackage[letterpaper,top=1in,bottom=1in,left=1in,right=1in,marginparwidth=1.75cm]{geometry}

\usepackage[T1]{fontenc}%for Leszek' name
\usepackage{amsmath}
\usepackage{amsthm}
\usepackage{tcolorbox}
\usepackage{multicol}

\newtheorem{theorem}{Theorem}
\newtheorem{lemma}{Lemma}
\newtheorem{fact}{Fact}
\newtheorem{corollary}{Corollary}

\title{Improving Efficiency in Near-State and State-Optimal Self-Stabilising Leader Election Population Protocols}
\author{Leszek G\k asieniec, Tytus Grodzicki, and Grzegorz Stachowiak%
\footnote{Supported by the National Science Center, Poland (NCN), grant 2020/39/B/ST6/03288.}
}
\date{}

\begin{document}
%\nolinenumbers

\maketitle

\begin{abstract}
    We investigate {\em leader election problem} via {\em ranking} within self-stabilising population protocols. In this scenario, the agent's state space comprises $n$ {\em rank states} and $x$ {\em extra states}. The initial configuration of $n$ agents consists of arbitrary arrangements of rank and extra states, with the objective of self-ranking. Specifically, each agent is tasked with stabilising in a unique rank state {\em silently}, implying that after stabilisation, each agent remains in its designated state indefinitely.

    %Known solutions to self-stabilizing leader election are achieved by reducing the problem to ranking.
    
    The ranking problem stands as a pivotal challenge in self-stabilising population protocols, with its resolution resulting in leader election and requiring a minimum of $n$ states~\cite{CIW12}. Within this framework, a self-stabilising ranking protocol is deemed {\em state-optimal} if it exclusively employs rank states. It is assumed to be {\em near-state optimal} if it incorporates $O(\log n)$ extra states, denoted as $x=O(\log n)$. Additionally, a configuration of agents' states is referred to as {\em $k$-distant} (from the final configuration) if exactly $k$ rank states are not used in the initial setup.

    A generic state-optimal ranking population protocol ${\mathcal A}_G$ demonstrates silent self-stabilisation within a time frame of $\Theta(n^2)$ with high probability (whp). It is established in~\cite{BCC+21,DS18} that any self-stabilising silent leader election protocol demands $\Omega(n)$ expected time. Furthermore,~\cite{BCC+21} introduces silent ranking protocols with expected times of both $O(n)$ and $O(n\log n)$ whp, each utilising $x=\Omega(n)$ extra states,
    recently improved to $x= O(\log^2 n)$ extra states in~\cite{BEG+25}.
    However, the exact complexity of near-state and state-optimal self-stabilising ranking remains elusive. Particularly, questions whether state-optimal self-stabilising ranking is viable within a time frame of $o(n^2)$ whp, and the exploration of the trade-off between the number of extra states and stabilisation time, remain unanswered.

    %In this paper, 
    We present several new self-stabilising ranking (and in turn leader election) protocols, greatly enriching our comprehension of these intricate problems.
    All protocols ensure self-stabilisation time with high probability (whp), defined as $1-n^{-\eta},$ for a constant $\eta>0.$
    We delve into three scenarios, from which we derive {\em stable} (always correct), either state-optimal or nearly state-optimal, silent ranking protocols that self-stabilise within a time frame of $o(n^2)$ whp, including:
    \begin{enumerate}
    \item In Section~\ref{s:k-distant}, we show that a novel concept of {\em an agent trap} admits a state-optimal ranking self-stabilising in time $O({\rm\/min}(kn^{3/2},n^2\log^2 n)),$ for any $k$-distant starting configuration.
    \item We also show that one extra state ($x=1$) enables a ranking protocol that self-stabilises in time $O(n^{7/4}\log^2 n)=o(n^2)$, regardless of the initial configuration, see Section~\ref{s:one}.
    \item Finally, in Section~\ref{s:logn}, we show that extra $x=O(\log n)$ states admit self-stabilising ranking with the best currently known time $O(n\log n)$, with whp and silent guarantees imposed.
\end{enumerate}

%Due to space constraints, all missing proofs can be found in Section~\ref{last}.

\end{abstract}

%\begin{CCSXML}
%	<ccs2012>
%		<concept>
%			<concept_id>10003752.10003809</concept_id>
%			<concept_desc>Theory of computation~Design and analysis of algorithms</concept_desc>
%			<concept_significance>500</concept_significance>
%		</concept>
%		<concept>
%			<concept_id>10003752.10010061</concept_id>
%			<concept_desc>Theory of computation~Randomness, geometry and discrete structures</concept_desc>
%			<concept_significance>500</concept_significance>
%		</concept>
	%
%	</ccs2012>
%\end{CCSXML}

%\ccsdesc[500]{Theory of computation~Design and analysis of algorithms}
%\ccsdesc[500]{Theory of computation~Randomness, geometry and discrete structures}    

%\keywords{Population protocols, Self-stabilisation, Ranking problem}


\maketitle

\section{Introduction}

The {model} of {\em population protocols} originates from the seminal paper of 
Angluin et al.~\cite{DBLP:conf/podc/AngluinADFP04}, and it is used to study the power of {\em pairwise interactions} between simple indistinguishable entities referred to as {\em agents}. 
In this model, each agent is equipped with a limited storage, a single state drawn from the predefined state space. 
It is often assumed that the state space of agents is fixed, as such protocols are independent 
from the size $n$ of (the number of agents in) the population. 
However, in the {\em ranking problem} one has to utilise $n$ {\em rank states} in order to distinguish 
between all agents. In addition, some extra $x$ states can be used to perform ranking more efficiently.
We say that a self-stabilising ranking protocol is {\em state-optimal}, if it utilises only rank states, 
and {\em nearly state-optimal}, if it uses also $O(\log n)$ extra states, i.e., $x=O(\log n)$.

In the original model of population protocols, it is assumed that a protocol begins in the predefined initial configuration of agents' states, representing the input, and stabilises in one of the final configurations of states, representing the solution to the considered problem.
In contrast, in self-stabilising protocols, the starting configuration is assumed to be an arbitrary configuration of states drawn from the state space, which in the considered problem consists of rank and extra states.
We also distinguish $k$-{\em distant configurations} of states, in which exactly $k$ rank states are missing, 
i.e., they are not occupied by any agent in this configuration.
%A population protocol is said to stabilise {\em silently}, if in the final configuration 
%all agents remain in their states indefinitely.

In the {\em probabilistic variant} of population protocols adopted here, 
in each step of a protocol the {\em random scheduler} selects an ordered pair of agents: the {\em initiator} and the {\em responder}, which are drawn from the whole population uniformly at random.
Note that the lack of symmetry in this pair is a powerful source of random bits utilised by population protocols.
In this variant, in addition to efficient  {\em state utilisation} one is also interested in the {\em stabilisation time} (often referred to as {\em parallel time}) defined as the number of interactions leading to the final configuration divided by the size of the population $n.$ 
We say that a protocol is {\em stable}, if it stabilises with the correct answer with probability 1, and is {\em silent}, if agents stop changing their states in a final configuration. 
All protocols presented in this paper are both {\em silent} and {\em stable}, and they guarantee stabilisation time with high probability (whp)
defined as $1-n^{-\eta},$ 
for a constant $\eta>0.$
Several lemmas in Section~\ref{s:one} utilise the Chernoff bound, see Section~\ref{s:app}, via Corollary~\ref{coro}.

\subsection{Related work}

%The ranking problem is adhered to {\em leader election}
%in self-stabilising population protocols, as any proper ranking affirms leader election in this model. 
%
%As indicated earlier, in self-stabilising protocols the ranking problem is closely related to the central problem of leader election, where in the final configuration a single agent 
%is in a {\em leader} state 
%and all other agents adopt 
%the {\em follower} state. 
In the standard model (with predefined starting configuration) of population protocols, the results in~\cite{DBLP:conf/wdag/ChenCDS14, DBLP:conf/soda/Doty14} laid down the foundation for the proof that leader election cannot be solved in a sublinear time with agents utilising a fixed number of states~\cite{DS18}.
In further work~\cite{DBLP:conf/icalp/AlistarhG15}, Alistarh and Gelashvili studied the relevant upper bounds including 
a new leader election protocol stabilising in time $O(\log^3 n)$ assuming $O(\log^3 n)$ states per agent.
Later, Alistarh {\em et al.}~\cite{DBLP:conf/soda/AlistarhAEGR17} considered more general trade-offs between the number of states and the time complexity of stabilisation.
In particular, they proposed a separation argument distinguishing between {\em slowly stabilising} population protocols 
which utilise $o(\log\log n)$ states and {\em rapidly stabilising} protocols relying on $O(\log n)$ states per~agent. 
This result coincides with another fundamental result by Chatzigiannakis {\em et al.}~\cite{DBLP:journals/tcs/ChatzigiannakisMNPS11} 
stating that population protocols utilising $o(\log n)$ states are limited to semi-linear predicates,
while the availability of $O(n)$ states (permitting unique identifiers) admits computation of more general symmetric predicates.
Further work proposes leader election in time 
$O(\log^2 n)$ w.h.p. and in expectation utilising $O(\log^2 n)$ states~\cite{DBLP:conf/podc/BilkeCER17}. 
The number of states was later reduced to $O(\log n)$ by
Alistarh et al. in \cite{DBLP:conf/soda/AlistarhAG18} and by Berenbrink et al. in \cite{DBLP:conf/soda/BerenbrinkKKO18} through the application of two types of synthetic coins.
In more recent work G\k asieniec and Stachowiak reduce utilisation of states to $O(\log\log n)$ while preserving the time complexity $O(\log^2 n)$ whp~\cite{DBLP:journals/jacm/GasieniecS21}.
The high probability can be traded for faster leader election in the expected parallel time $O(\log n\log\log n)$~\cite{DBLP:conf/spaa/GasieniecSU19}. This upper bound was recently reduced to the optimal expected time $O(\log n)$ by Berenbrink {\em et al.} in~\cite{DBLP:conf/stoc/BerenbrinkGK20}. 
%One of the main open problems in the area is to establish whether one can elect a single leader in time $o(\log^2 n)$ whp while preserving the optimal number of states $O(\log\log n).$ 

The fact that at least $n$ states on the top of the knowledge of the exact value of $n$ are needed in self-stabilising leader election (as in the ranking problem) was first observed in~\cite{CIW12}. 
One could use an upper bound on $n$ instead, if {\em loose-stabilisation} is permitted, i.e., when the selected leader remains for an extended period of time, but then it has to be recomputed again~\cite{Sudo+20}.
Another example refers to a recent study on $O(\log n)$-state loosely-stabilising phase clocks with application to an adaptive variant of the majority problem, see~\cite{BBH+22}. 
An alternative line of attack is to supply population protocols with extra features, where in the context of leader election one can find work on population protocols with an oracle~\cite{BBB13, FJ06}, mediated population protocols~\cite{MOK+12}, and population protocols with interactions along limited in size hyper-edges~\cite{XYK+13}. 
A different strand of work refers to communication networks limited to constant degree regular graphs including rings, in which state requirements are much lighter~\cite{AAF+05,CC19,CC20,YSM20,YSO+23}. 

The consideration of $k$-distant configurations in this paper has ties to $k$-stabilising protocols, which are self-stabilising protocols characterised by a known upper bound $k\leq n$ on the number of faults~\cite{BGK99}.
Numerous works, e.g., ~\cite{BGK99,GX99,KP97,KP95}, observed that systems 
recover more rapidly when subjected to fewer faults. 
The most efficient {\em $k$-linearly adaptive} protocols require only $O(k)$ transitions to stabilise, see~\cite{BDH06}. 
Some lower bounds for such protocols have been explored in~\cite{GT02}.
Due to the random nature of interactions and a %severely 
constrained state space, our self-stabilising ranking protocol for $k$-distant configurations, while bearing similarities to $k$-stabilising protocols, still necessitates engagement from the entire population, resulting in a stabilisation time of $O({\rm\/min}(kn^{3/2},n^2\log^2 n))$. However, the value of $k$ in our protocol does not have to be predetermined.

In the ranking problem, which is the focus of this work, the agent's state space consists of 
$n$ {\em rank} states and $x$ {\em extra} states. The initial configuration of $n$ agents is 
an arbitrary arrangement of rank and extra states, with the goal of achieving self-ranking. 
Specifically, each agent must stabilise in a {\em unique} rank state silently, meaning that 
after stabilisation, each agent remains in its designated state indefinitely.
The ranking problem is closely related to the {\em naming (labelling)} problem, where each agent is 
assigned a unique name (or label), as studied earlier in \cite{Bea12,Mic13,Bur19,Gas24,Gas25}  within the context 
of population protocols. The key difference is that, after ranking, any pair of agents knows 
how many agents with intermediate ranks separate them. Additionally, due to the restrictions on 
final states, computing the ranking is generally more challenging.
In this context, the primary focus of \cite{Bea12,Bur19} is on the feasibility of ranking and naming, 
while \cite{Mic13} examines the time required for the population to stabilise. A more comprehensive study 
on the state and parallel efficiency of naming is presented in \cite{Gas25}, where the authors explore 
trade-offs between the number of names, available states, and stabilisation time. Another recent work, 
\cite{Gas24}, investigates a specific naming scheme (by consistent coordinates) in the context of distributed 
localization within {\em spatial} population protocols.
A related {\em renaming} (namespace reduction) problem, introduced in \cite{Att90}, has been studied in terms of solvability and complexity in asynchronous environments, e.g., in \cite{Ebe98,Afe99,Her99,Att01}, and with respect to time complexity in the shared-memory model, e.g., in \cite{Bor93,Pan98,Chl08,Ali14}.

In the direct context of our work, in~\cite{BCC+21,DS18}, the authors show that any silently self-stabilising leader election protocol requires $\Omega(n)$ expected time. 
In~\cite{BCC+21}, silent self-stabilising ranking (and in turn leader election) protocols achieve expected time $O(n)$ and $O(n \log n)$ whp, using $\Omega(n)$ extra states, recently improved in~\cite{BEG+25} to $O(\log^2 n)$ extra states.
However, the only known state-optimal ranking protocol $\mathcal{A}_G$ 
with the state space 
$\{0,\dots,n-1\}$ 
based on a single rule $i+i\rightarrow i+(i+1{\rm\ mod\ } n)$,
i.e., if two agents in state $i$ interact, the responder transitions to state $(i+1\text{ mod } n)$.
This protocol
has stabilisation time $\Theta(n^2).$

\section{Basic observations, tools and probabilistic guarantees}\label{s:basic} 

In this work, we focus on self-stabilising population protocols for the ranking problem. In~such protocols, the starting configuration contains an arbitrary combination of rank (and extra) states. In the final configuration, each agent has to stabilise in a unique rank state.
%
%Our protocols are {\em stable}, i.e., they always compute correct solutions, and the stabilisation times are guaranteed with high probability. 

We say that a state is {\em occupied} if there is an agent in this state in the current configuration.
We note that for any (rank) state $s$ in a state-optimal ranking protocol, there exists a rule $(s,s) \rightarrow (s',s'')$, where $s''\neq s$. This rule serves to offload states occupied by more than one agent. Furthermore, these are the only rules permitted; otherwise, the protocol would never stabilise.
Thus, every state-optimal ranking protocol has a transition function formed exactly of $n$ rules. For instance, in the generic ranking protocol $\mathcal{A}_G$ with the state space $0, \dots, n-1$, the only rules are of the form $i+i\rightarrow i+(i+1{\rm\ mod\ } n)$, for all $0 \leq i \leq n-1$. This constraint leaves very little room for manoeuvre in the design of such protocols.
In order to outwit this restriction 
%in state-optimal and to great extent in nearly state-optimal ranking protocols 
we design and utilise several combinatorial tools including: an {\em agent trap} introduced in this section, a {\em ring of traps} in Section~\ref{s:k-distant}, a {\em line of traps} in Section~\ref{s:one},  and 
{\em size-balanced} binary trees in Section~\ref{s:logn}.

\subsection{Agent trap}

We propose a novel concept called an {\em agent trap}, where a trap of size $m+1$ is formed by $m+1$ rank states supported by $m+1$ rules in the transition function, allowing indefinite entrapment of $m+1$ agents. Let $0, \dots, m$ represent the states used in the trap, where $0$ is the gate state, and all remaining $m$ states are referred to as {\em inner states}. Each inner state $0 < i \leq m$ is associated with rules ${\mathcal R}_i: i+i \rightarrow i+(i-1)$, for $i=1,\dots,m,$ and the {\em gate state} is associated with the rule ${\mathcal R}_g: 0+0 \rightarrow m+Y$, where $Y$ is a state not belonging to this trap. It's important to note that such a trap degenerates to a single state when $m=0$. In our protocols, $Y$ is either the gate state of another trap or the unique extra state utilised in nearly state-optimal ranking protocols described in Section~\ref{s:one}.
Generally, the inner states are employed to capture agents, while the gate state is responsible for reducing an excess of entrapped agents and for accepting new agents into the trap. 

An unoccupied inner state is called a {\em gap}. 
\begin{fact}
    Once a gap state becomes occupied, it remains occupied indefinitely.
\end{fact}

%We say that a trap has $k$ {\em gaps} if $k$ inner states in this trap are not occupied by any agent. 
We say that a trap is {\em saturated} if there are no gaps in this trap.
We also say that a trap is {\em full} if it is saturated and at least $m+1$ agents occupy (states in) this trap.

\begin{fact}~\label{f:2k}
%    Arrival of at most $2d$ new agents at the gate of a trap with $d$ gaps guarantees saturation of this gap. 
To saturate a trap with $d$ gaps, where the gate ejects every other agent, the total number of agents currently in the trap plus those expected to reach its gate state must be at least $2d$.

%As the gate ejects every other agent from the trap, to saturate a trap with $d$ gaps, it is sufficient for the combined number of agents currently in and those expected to arrive at its gate state to be at least $2d$.
\end{fact}
%
%\begin{proof}
%    A half ($k$) of the newly arrived agents will enter state $g+m$ and further descend towards the existing gaps.
%    This number can be smaller than $k,$ if some gaps are filled in (become occupied) by agents already present in the trap. 
%\end{proof}
%
\iffalse
\begin{proof}
    Rule ${\mathcal R}_g$ guarantees that half of the $2d$ agents in gate state $g$ adopt state $g+m,$ and these agents gradually descend to all gaps using rules ${\mathcal R}_i,$ for $i=m,\dots,2$.
\end{proof}
\fi

\begin{fact}\label{f:full}
    Once a trap becomes full, it remains full indefinitely.
\end{fact}
%
\iffalse
\begin{proof}
    Since a full trap can loose an agent only by application of rule ${\mathcal R}_g$ when at least two agents occupy the gate state, in such case the number of agents in the trap is at least $m+2.$
    After rule ${\mathcal R}_g$ is applied one of these agents adopts state $g+m.$
    Thus, the number of agents in the trap remains at least $m+1.$
\end{proof}
\fi

We say that a trap has a {\em surplus} $l\ge 0$ if $m+l+1$ agents occupy states of this trap.

\begin{lemma}\label{l:surplus}
    Assume $m>n^{\varepsilon},$ for some constant $\varepsilon>0.$ A trap with a surplus $l>0$
    \begin{enumerate}
        \item releases (and returns back to the trap) at least $\lfloor\frac{l+1}{2}\rfloor\ge 1$ agents in time $mn$ whp, and
        \item releases at least $l$ agents in time $mn(\lceil\log (l+1)\rceil+1)$ whp.
    \end{enumerate}
    Moreover, if $l=0$ and the trap is full, then in time $mn$, an agent arrives at the gate state.
\end{lemma}
\begin{proof}
Consider a trap with a positive surplus $l$. Since such a trap may have gaps (unoccupied states), we can identify two groups of agents descending (possibly from as far as state $g+m$) toward the gate state $g$ within this trap. The first group includes all agents that reach state~$g$. 
The second group is made up of the remaining agents which have filled the gaps inside the trap.
%The second group consists of any other descending agents that have occupied gaps within this trap.
%
While agents in these two groups may have different destinations, they all utilise rules ${\mathcal R}_i$ to reach their respective destinations, with each agent using such a rule at most $m$ times. Furthermore, we can assume that for any descending agent, all inner states along the path to the destination are already occupied. Otherwise, the agent would stop earlier at a gap.
To be more precise, an agent that occupies a gap first may later be displaced from that state by another agent. However, as agents are indistinguishable, we can interpret such an event as the later arrival moving rather than the reverse.
%
%This implies that during every interaction of an agent targeting the gate state, the probability of the agent moving one state closer to its destination (i.e., a transition rule ${\mathcal R}_i$ being applied) is at least $\frac{2}{n^2}$.
The probability that, in a given interaction, an agent heading towards the gate state will take a step forward is at least as high as the probability of interaction with the agent occupying the inner state, i.e., $\frac{2}{n^2}$.
Therefore, over $0.75mn^2$ interactions, the expected number of movements of this agent along descending states within the trap is at least $\mu=1.5m$.
By the Chernoff bound, the probability that this number is smaller than $m=\frac{2}{3}\mu$ is at most $e^{-\mu/18}=e^{-m/12}<n^{-\eta-2}$, assuming $m>n^{\varepsilon}$.
After the subsequent $0.75mn^2$ interactions, $\lfloor\frac{l+1}{2}\rfloor$ agents visiting the gate state $g$ exit this trap, with the same number being transferred to state $g+m$. Also, the probability that an agent vacates state $g$ in an interaction is at least $2/n^2$. Thus, the probability that this agent remains in state $g$ during $\frac{1}{4}mn^2$ interactions is less than $(1-2/n^2)^{mn^2/4}<e^{-m/2}<n^{-\eta-2}$.
Finally, by the union bound, during $mn^2$ interactions, which correspond to a time of $mn$, at least $\lfloor\frac{l+1}{2}\rfloor$ agents leave the trap.
%The number of interactions $mn^2$ is equivalent to the time $mn$.

When $l=0$, the same process ensures that an agent arrives at the gate state.
    
    After at least $l+1$ agents visiting the gate state during the considered $mn^2$ interactions are released whp, at most $\lceil\frac{l+1}{2}\rceil$ agents will visit the gate state again during  the subsequent $mn^2$ interactions whp. By iterating this reasoning, after the subsequent $mn^2$ interactions, we are left whp with at most $\lceil\frac{l+1}{4}\rceil$ agents that are yet to visit the gate state. Thus, whp after  $mn^2\lceil\log (l+1)\rceil$ interactions, we are left with exactly one agent that is yet to visit the gate state. This agent arrives at the gate state whp after the subsequent $mn^2$ interactions, which correspond to a time of $mn$.
\end{proof}

Finally, we say that a trap occupied by exactly $m+1$ agents is {\em almost stabilised}, if it is %saturated, but one of the inner states is occupied by two agents, i.e., 
full and the gate state is empty, and %a trap is
{\em fully stabilised} if each state in the trap is occupied by exactly one agent.

\subsection{Tidy configurations}

In Sections \ref{s:k-distant} and \ref{s:one}, we investigate ranking protocols based on multiple traps interconnected by their gate states. In such systems, we define \textit{tidy configurations}, where within each trap, the \textit{overloaded} inner states occupied by more than one agent have higher indices than all gaps (unoccupied inner states) in this trap. The following lemma holds.

\begin{lemma}\label{l:tidy}
    The configuration of states becomes and remains tidy 
    after time $mn$ whp.
\end{lemma}
\begin{proof}
By %the proof of 
Lemma~\ref{l:surplus} all agents that do not fill the gaps, get to the gate in time $0.75mn<mn$ whp.
Once it happens, the resulting configuration is tidy
%    Using an analogous argument to the proof of Lemma~\ref{l:surplus} and the union bound, one can demonstrate that in each trap, all excess agents either fill in existing gaps or descend to the gate state after $O(mn^2)$ interactions whp. And when this occurs, 
and all subsequent configurations are also tidy. 
    
\end{proof}

%\subsection{Tidy configurations}

%Consider a protocol utilising traps. Recall that agents can only enter traps via gate states. 

\section{State-optimal k-distant ranking}\label{s:k-distant}

In this section, we present and analyse a new state-optimal ranking protocol stabilising silently in time $O({\rm\/min}(kn^{3/2},n^2\log^2 n)),$ for any $k$-distant starting configuration.
The argument is split into two cases with small values of $k$ (Section~\ref{ss:km}), and arbitrarily large $k$ (Section~\ref{ss:arbitrary}).

\subsection{Ring of traps}

In Section~\ref{s:basic}, we introduced an agent trap which admits indefinite entrapment of agents.
Here we extend this construction to a {\em $(m,m+1)$-ring-of-traps} formed of $m$ traps of size $m+1,$ where $(a,0)$ constitute the gate states in each trap $a=0,\dots,m-1,$ and $(a,b),$ for $b=1,\dots,m,$ refer to the inner states. % $b$ in trap $a$. 
Below we assume that $n=m(m+1),$ however, one can reduce some traps to less than $m+1$ states
to accommodate any other value of $n.$

The inner states $(a,b)$ use rules ${\mathcal R}_i$ supporting movement of excess agents %towards the gate state of
along the trap, where
\[(a,b)+(a,b)\rightarrow(a,b)+(a,b-1), {\rm\ for\ }b>0,
\]
while the use of the gate state the rule ${\mathcal R}_g$ results in sending one agent to the top inner state and expediting the other to the gate state of the next trap on the ring, where  
\[(a,0)+(a,0)\rightarrow(a,m)+((a+1){\rm\ mod\ }m,0).
\]

Recall that in a $k$-distant configuration, there are exactly $k$ states that are unoccupied. 
In a ring of traps, some of these missing states correspond to gaps, while others represent currently unoccupied gate states. It's important to note that the number of gaps is monotonically decreasing; that is, the inner states become increasingly occupied. However, the number of unoccupied gate states may vary over time. In our argument, we first address filling of all gaps and then focus on the full stabilisation of all traps in the ring.

\subsection{Ranking for k-distant initial configurations}\label{ss:km}
Although filling gaps occurs simultaneously in real terms, for small values of $k$, we first establish an upper bound on the time required to fill a single gap and then multiply this bound by the number $k$ of gaps.
%While in real terms reduction of gaps is done simultaneously,
%in our argument for small values of $k$ we provide an upper bound for a single gap reduction, and then multiply this bound by the number of at most $k$ gaps. 
We start with the following lemma.
\begin{lemma}
\label{f:saturation}
    The time needed to stabilise any $k$-distant configuration is $O(kn^{3/2})$ whp. 
\end{lemma}
\begin{proof}

We say that a trap is \emph{flat} if no internal state of the trap is occupied by more than one agent. Let \( k_1 \) denote the number of flat traps with unoccupied gate states, where each such trap is assigned a weight of 1. Additionally, let \( k_2 \) represent the total number of gaps across all traps, with each gap assigned a weight of 2. Note that the \emph{total weight} $K$ of such a configuration satisfies \(K= k_1 + 2k_2 \leq 2k \). 

We observe that the total weight $K$ does not increase when transitioning to the next configuration. Specifically, \( k_2 \) cannot increase, since the number of gaps decreases over time. Furthermore, \( k_1 \) can only increase in two cases: (1) when a trap becomes flat, or (2) when the gate state of a flat trap becomes unoccupied.

In case (1), there must be a single internal state \( i \geq 1 \) occupied by exactly two agents, which interact according to the rule \( i+i \rightarrow i+(i - 1) \).
If \( i > 1 \), the gap state \( i - 1 \) becomes occupied, \( k_2 \) is reduced by 1, \( k_1 \) is increased by 1, and consequently, \( K \) is reduced by 1. If \( i = 1 \), both \( k_1 \) and \( K \) remain unchanged, since the number of flat traps has not changed.

In case (2), when the gate state of a trap becomes unoccupied, one agent adopts state \( m \) in this trap. If after this transition the trap is flat, i.e., if state \( m \) was unoccupied before this transition, \( k_1 \) is increased by 1 but \( k_2 \) is decreased by 1, resulting in reduction of $K$ by 1. Alternatively, if the trap is no longer flat, i.e., if state \( m \) is now occupied by 2 agents, both \( k_1 \) and \( K \) remain unchanged, since the number of flat traps with empty gates has not changed.

Thus, we can conclude that the total weight $K$ decreases monotonically. In the following, we focus on the rate at which the total weight $K$ is reduced over time.

According to Lemma~\ref{l:tidy}, after time $O(n^{3/2})$ whp any initial configuration becomes tidy. 
We now show that, within \( O(n^{3/2}) \) time, the total weight \( K \) of a tidy configuration is reduced by at least 1, provided that \( K > 0 \).

Assume first that there is a trap with a gap below a state occupied by at least two agents. In this case, some agents will descend and fill at least one gap within time \( mn = O(n^{3/2}) \) whp, as argued in Lemma \ref{l:surplus}. Consequently, \( K \) is reduced by at least 1.

Otherwise, if no such traps exist, a sequence of traps \( a, (a+1), \dots, (a+c) \) modulo \( m \) can be found, such that trap \( a \) has a surplus, traps \( a+1 \) through \( a+c-1 \) are almost or fully stabilised, and trap \( a+c \) is not full. 
By Lemma~\ref{l:surplus}, trap \( a \) releases at least one agent in time \( m = O(n^{3/2}) \) whp. Additionally, all intermediate traps \( a+1 \) through \( a+c-1 \) stabilise with their gates occupied by agents within time \( O(n^{3/2}) \) whp. This ensures that the sequence of agent transfers through the gates of traps \( a \) through \( a+c \) is completed in time \( O(n^{3/2}) \) whp, guaranteeing that a new agent reaches the gate of trap \( a+c \). 

If this new agent encounters an unoccupied gate state of trap $a+c$ and this trap is flat, \( k_1 \) is reduced by 1 instantly. Otherwise, i.e., trap $a+c$ is not flat (and configuration is tidy), after time \( O(n^{3/2}) \) whp, at least one gap in this trap gets filled in reducing $k_2$ by at least 1. 

If, on the other hand, the gate state of trap $a+c$ was occupied, then in time \( O(n\log n) \) one of the agents from the gate state moves to state $m$ of this trap. If state $m$ was unoccupied before this transition, the value of $k_2$ is reduced by 1 instantly. Otherwise, the highest gap in this trap is filled in time $O(n^{3/2})$ whp, also reducing $k_2$ by 1. Both actions may result in increasing $k_1$ by 1 (the trap becomes flat), which gives reduction of $K$ by at least 1. 

When the total weight \( K \) reaches 0, stabilisation is achieved by filling all the gates in the traps that are already almost stabilised. This occurs in time \( O(kn^{3/2}) \) whp.

\end{proof}

\subsection{Ranking for arbitrary initial configurations}\label{ss:arbitrary}

For the completeness of presentation we also consider arbitrarily large values of $k.$ In this case we utilise greater parallelism of agents' transfers guaranteed by the second case of Lemma~\ref{l:surplus}.

\begin{lemma}\label{l:arbit}
    For an arbitrary initial configuration, the $(m,m+1)$-ring-of-traps protocol stabilises in time $O(n^2\log^2 n)$ whp.
\end{lemma}
\begin{proof}

By Lemma~\ref{l:tidy}, after time $O(n^{3/2})$, the configurations of states are tidy. We will prove that any trap becomes full after time $O(n^2 \log^2 n)$ whp. Without loss of generality, we will demonstrate this for trap 0. By symmetry, the result holds for any trap, and by the union bound, it will hold for all traps simultaneously.

Assume that in a tidy configuration, trap 0 contains \( m - d \leq m \) agents in its inner states. For saturation, such a trap requires the same number of agents to arrive at the gate as a trap with \( d \) gaps. According to Fact~\ref{f:2k}, this number of agents is \( 2d \). Furthermore, the arrival of one additional agent causes the trap to become full. This process is completed in time \( O(n^{3/2}) \) whp, after all these agents have arrived at the gate of trap 0.

Let there be $m-d$ agents in inner states of trap 0
and $\gamma_0$ agents in the gate state.
We will show that at least $d'=d-\gamma_0+1$ agents will arrive at this gate state within time $O(n^2\log n)$ whp.
This will result in at least $\lceil d/2\rceil$ new agents in the inner states of trap 0 in time $O(n^2\log n)$~whp.

If \( d' \le 0 \), then all required agents are already in the gate state.  
Otherwise, let the traps \( 0, 1, 2, \ldots \) contain \( n_0, n_1, n_2, \ldots \) agents, respectively.  
There is an equality \( n_0 + n_1 + \cdots + n_{m-1} = n = m(m+1) \).  
Since \( n_0 = m - d + \gamma_0 = m + 1 - d' \), we have  
$n_1 + \cdots + n_{m-1} = (m-1)(m+1) + d'.$
From the second point of Lemma~\ref{l:surplus}, after time \( O(mn \log n) \), at least \( n_1 - (m+1) \) agents will transition to trap 2 whp.  
Note that if \( n_1 - (m+1) \le 0 \), then no agents are required to move; and if \( n_1 - (m+1) > 0 \), then this value is the surplus to which Lemma~\ref{l:surplus} refers.  
By the same principle, in the next period of length \( O(mn \log n) \), at least \( n_1 + n_2 - 2(m+1) \) agents will transition to trap 3 whp.  
After the next period of length \( O(mn \log n) \), at least \( n_1 + n_2 + n_3 - 3(m+1) \) agents will go to trap 4 whp, and so on.  
Finally, after a total time of \( O(n^2 \log n) \), at least \( n_1 + \cdots + n_{m-1} - (m-1)(m+1) = d' \) agents will reach the gate state of trap 0 whp.

Since at least \( \lceil d/2 \rceil \) new agents enter the inner states of trap 0, at least \( m - \lfloor d/2 \rfloor \) agents will be in the inner states of trap 0 after time \( O(mn \log n) \) whp.  
After repeating this process \( \log d \) times, we obtain a configuration with at least \( m \) agents in the inner states of trap~0.  
One further repetition will cause trap 0 to become full.  
This results in a total time of \( O(n^2 \log^2 n) \) whp.

Since at least $\lceil d/2\rceil$ new agents enter the inner states of trap 0, then at least $m-\lfloor d/2\rfloor$ agents are in inner states of trap 0
after time $O(mn\log n)$ whp.
After repeating this process $\log d$ times, we will obtain a configuration with at least $m$ agents in the inner states of trap~0.
One further repetitions will cause trap 0 to become full.
This results in the total time $O(n^2\log^2 n)$ whp.
\end{proof}

By combining Lemma \ref{f:saturation} and Lemma~\ref{l:arbit} we get Theorem \ref{l:kdist}.

\begin{theorem}\label{l:kdist}
    By solving ranking for a $k$-distant configuration, %with $k < n$, 
    we obtain a silent, self-stabilising leader election protocol with a running time of $O(\min(kn^{3/2}, n^2 \log^2 n))$ whp.
\end{theorem}

%\begin{theorem}\label{l:kdist}
%    For an arbitrary $k<n,$ the time needed to silently stabilise a $k$-distant configuration is \linebreak $O({\rm\/min}(kn^{3/2},n^2\log^2 n))$ whp.
%\end{theorem}

In particular, when $k=o(\sqrt{n})$ leader election protocol stabilises in time $o(n^2)$.

%In particular, when $k=o(\sqrt{n})$ {\em ring-of-traps} protocol stabilises in time $o(n^2)$.

\section{Ranking with one extra state}\label{s:one}

In this section, we demonstrate that the introduction of one additional state, i.e., when $x=1$, enables a reduction in the stabilisation time of ranking from $O(n^2)$ to $O(n^{7/4}\log^2 n)=o(n^2)$ for an arbitrary initial configuration of states.
This new 
%ranking 
protocol is based on the concept of a {\em line of traps}.

\subsection{Lines of traps}

A {\em line of traps} comprises a sequence of traps connected consecutively. In the new ranking protocol $m^2$ lines of traps are deployed, with each line containing $3m$ traps of size $m+1$.
Hence, $n=3m^3(m+1).$
%More formally, a line of traps is a subset of states denoted as $(a,b)$, where $a\in[1,3m],b\in[0,m]$,
%and the rules outlined in the transition function below.
More formally, a line of traps is a subset of states denoted as $(a,b)$, where $a\in[1,3m]$ and $b\in[0,m]$, subject to the rules outlined in the transition function below.
%governed by rules of transition function described below.
Similar to the ring of traps, within each trap $a=1,\dots,3m,$ state $(a,0)$ represents the gate state, while the remaining states $(a,b),$ for $b>0,$ denote the inner states.
For every line, the gate $(3m,0)$ functions as the {\em entrance gate}, while the gate $(1,0)$ acts as the {\em exit gate}.

The ranking protocol utilises also a single extra state $X$ (not belonging to any line of traps) which collects agents released by the exit gate in each line of traps.
%Agents can be sent from state $X$ to the entrance gate of any line at random time.
%
In Section~\ref{s:prot_def}, we describe a mechanism that utilises random interactions to transition agents in state $X$ to the entrance gate of any line.
%
%In Section~\ref{s:prot_def}, we describe a mechanism that uses random interactions involving agents in state $X$ to transition them to the entrance gate of any line.
%
%We describe in Section~\ref{s:prot_def} a mechanism that uses random encounters with other agents to send agents from state $X$ to the entrance gate of any line.
%at a random time.
We have the following transition rule for inner states in each trap:
\[
(a,b)+(a,b)\rightarrow (a,b)+(a,b-1), \mbox{ for } b>0.
\]
The following transition rule moves agents between two consecutive traps on the line:
\[
(a,0)+(a,0)\rightarrow (a,m)+(a-1,0), \mbox{ for } a>1.
\]
Finally, the (exit) gate of the last trap releases agents to the extra state $X$:
\[
(1,0)+(1,0)\rightarrow (1,m)+X.
\]

The arrangement of agents in a line $l$ %at a given moment of computation
is called a configuration~\( C_l \).
It describes the numbers of agents in all states at a given moment.
%It describes the numbers of agents in all states along the line.
Denote by $|C_l|$ the number of agents present in configuration $C_l$.
We say that a line is {\em full}, when all its traps are full.
By Lemma \ref{l:tidy}, after time $O(mn)$ all configurations $C_l$ remain tidy whp. Due to this fact, we assume that all configurations analysed in this section are tidy.
For the sake of the analysis, a tidy configuration of agents in a line can be described by two vectors.

\begin{itemize}
\item
In the {\em inner vector} \( \beta(C_l) = (\beta_{3m}(C_l), \beta_{3m-1}(C_l), \ldots, \beta_2(C_l), \beta_1(C_l))
\),
the value \( \beta_a(C_l) \) represents the number of agents that are in the inner states of trap \( a \).
\item
In the {\em gate vector} \( \gamma(C_l) = (\gamma_{3m}(C_l), \gamma_{3m-1}(C_l), \ldots, \gamma_2(C_l), \gamma_1(C_l)),
\)
the value \( \gamma_a(C_l) \) represents the number of agents that are in the gate state of trap \( a \).
\end{itemize}

In Lemma~\ref{unique}, we prove that if no agents arrive at the gate of line $l$ during the execution of the protocol, the number of agents released by this line before reaching a silent configuration depends solely on the initial configuration $C_l$. We call this number of released agents the {\em surplus} and denote it by $s(C_l)$.

\begin{lemma}\label{unique}
Assume that the protocol executed on line \( l \) starts in a  configuration \( C_l \) with inner vector \( \beta(C_l) \) and gate vector \( \gamma(C_l) \), and that, during the execution of the protocol, no new agents arrive at the entrance gate of line \( l \). 
Let \( s(C_l) \) denote the number of agents released by the line before eventually reaching a silent configuration \( \overline{C_l} \), with vectors \( \beta(\overline{C_l}) = (\beta_{3m}(\overline{C_l}), \ldots, \beta_1(\overline{C_l})) \) and \( \gamma(\overline{C_l}) = (\gamma_{3m}(\overline{C_l}), \ldots, \gamma_1(\overline{C_l})) \). 
The content of the vectors \( \beta(\overline{C_l}) \), \( \gamma(\overline{C_l}) \), and the value of \( s(C_l) \) depend solely on the initial configuration \( C_l \).

\end{lemma}

\begin{proof}

    Let $\beta_a=\beta_a(C_l)$ and $ \gamma_a=\gamma_a(C_l).$ Moreover, for $a=0,\dots,3m-1,$ $x_a$ be the number of agents that trap $a$ receives from
    trap $a+1$, where $x_{3m}=0$ and $x_0=s(C_l)$ is the number of agents released from the line.
    We prove that sequences $x_a,\beta_a(\overline{C_l}),\gamma_a(\overline{C_l})$, and the value of $s(C_l)$ can be expressed in terms of $\beta(C_l)$ and $ \gamma(C_l)$.
    
    The proof is done by induction on the trap index $a$, descending from $a=3m$ to $1$.

Let us consider a fixed trap $a$.
Define $y_a=x_a+\gamma_a$ as the total number of agents visiting gate state $(a,0)$. Thus, the number of agents that enter trap $a$ during their initial interactions at this gate is $\lfloor y_a/2\rfloor$. We will consider two cases.

    If $\beta_a+\lfloor y_a/2\rfloor\leq m$, then $\beta_a(\overline{C_l})=\beta_a+\lfloor y_a/2\rfloor$. Additionally, $x_{a-1}=\lfloor y_a/2\rfloor$ agents are transferred from $(a,0)$ to $(a-1,0)$ (or to state $X$ if $a=1$). So $\gamma_a(\overline{C_l})=y_a - 2 \lfloor y_a/2 \rfloor$. In this case, no agents return to gate $(a,0)$ from inner states of trap $a$.

    If $\beta_a+\lfloor y_a/2\rfloor > m$, then $\beta_a(\overline{C_l})=m$ and $\gamma_a(\overline{C_l})=1$. Trap $a$ releases all agents not used to saturate it. Thus, $x_{a-1}=\beta_a+y_a-m-1$ agents are transferred from $(a,0)$ to $(a-1,0)$ (or to state $X$ if $a=1$).
\end{proof}

Let us assume that at some moment a single trap contains $\beta$ agents in the inner states and $\gamma$ agents in its gate.
One can calculate the numbers $\alpha\in[0,m]$, $\delta\in\{0,1\}$, and $\rho$ of agents left in inner states, the gate and released from the trap respectively, after total stabilisation without adding new agents to this trap.
If $\beta+\lfloor\frac{\gamma}{2}\rfloor\le m$, then $\alpha=\beta+\lfloor\frac{\gamma}{2}\rfloor$, $\delta=\gamma\mod 2$, $\rho=\lfloor\frac{\gamma}{2}\rfloor$.
If $\beta+\lfloor\frac{\gamma}{2}\rfloor> m$, then $\alpha=m$, $\delta=1$, $\rho=\beta+\gamma-m-1$.
We can now define vectors consisting of values $\alpha, \gamma, \rho$ for all traps in the line.

\begin{itemize}
\item
In the {\em allocation vector} \( \alpha(C_l) = (\alpha_{3m}(C_l), \alpha_{3m-1}(C_l), \ldots, \alpha_2(C_l), \alpha_1(C_l)),
\)
where\\ \( \alpha_a(C_l) =\min\left\{\beta_a(C_l)+\lfloor\frac{\gamma_a(C_l)}{2}\rfloor,m\right\}\).
\item
In the {\em target gate vector} \( \delta(C_l) = (\delta_{3m}(C_l), \delta_{3m-1},(C_l) \ldots, \delta_2(C_l), \delta_1(C_l)),
\)
where $\delta_a(C_l)=\gamma_a(C_l)\mod 2$, for $\beta_a(C_l)+\lfloor\frac{\gamma_a(C_l)}{2}\rfloor\le m$ and $\delta_a(C_l)=1$ otherwise.
\item
In the {\em excess vector} \( \rho(C_l) = (\rho_{3m}(C_l), \rho_{3m-1}(C_l), \ldots, \rho_2(C_l), \rho_1(C_l)),
\)
where \( \rho_a(C_l) = \lfloor\frac{\gamma_a(C_l)}{2}\rfloor\), if $\beta_a(C_l)+\lfloor\frac{\gamma_a(C_l)}{2}\rfloor\le m$ and $\rho_a(C_l)=\beta_a(C_l)+\gamma_a(C_l)-m-1$ otherwise.
%Let $r(C_l)=\sum_a \rho_a(C_l)$.
\end{itemize}

Consider the excess vector $\rho(C_l)=(\rho_{3m},\ldots,\rho_1)$,
where each trap $a$ contains {\em excess} of
size $\rho_a$.
We interpret the current excess of each trap as the number of {\em tokens} associated with this trap.
The total number of tokens is $r(C_l)=\sum_a \rho_a(C_l)$.
%We place {\em tokens} corresponding to excess units on traps. Specifically, for each
%$a$ we put $\rho_a$ tokens on trap $a$.
Tokens corresponding to vector $\rho(C_l)$ are numbered from 1 to $r(C_l)$ in such a way that the lower the trap number, the lower the token number.
Consider a sequence of interactions where $C_l$ is the initial
configuration and agents can enter the line during the protocol.
Subsequent interactions can either move tokens towards traps with lower numbers or eliminate them.
Agents entering the line turn into new tokens, which receive
subsequent numbers greater than $r(C_l)$.
A movement of an agent from trap $a$ to $a-1$ decreases the value of $\rho_a$ by one.
If this is associated with an increase in the value of $\rho_{a-1}$, we say that a token {\em moves} from trap $a$ to $a-1$.
On the other hand, if $\rho_{a-1}$ remains the same, a token is {\em eliminated} in trap $a-1.$
Such an elimination is associated with a permanent increase in the number of agents in the trap $a-1$.
We assume that if a trap contains multiple tokens, then tokens with the lowest numbers are moved to the next trap or eliminated first.
Moving tokens this way maintains their ordering consistent with the order of traps. 
Note that the number of tokens can increase only by adding
agents to the entrance gate of the line.
In particular, if no agents enter a line
and configuration $C_l'$ occurs after configuration $C_l$,
then
$r(C_l')\le r(C_l)$.
If a token is eliminated or passes through the entire line, resulting in an agent being released from the line, we say that this token has been {\em handled}.

We define an order on the configurations in such a way that $C_l\le C'_l$ if for all $a$, the conditions 
$\alpha_a(C_l)\le \alpha_a(C'_l)$, and $\alpha_a(C_l)+\delta_a(C_l)\le \alpha_a(C'_l)+\delta_a(C'_l)$ hold simultaneously.
Also in addition to the surplus $s(C_l)$ for configuration $C_l$, one can define its {\em deficit} $d(C_l)=3m(m+1)-|\overline{C_l}|$ as the number of unoccupied states in $\overline{C_l}$.

\begin{lemma}\label{reduce_d}
Let $C_l$ be a line configuration with deficit $d(C_l)$.
If $C_l'$ is obtained by inserting $\min\{d(C_l)+m,2d(C_l)\}$ additional agents at the entrance gate of this line, % to the line in configuration $C_l$,
then $d(C_l')=0$.
\end{lemma}

\begin{proof}
Every other agent entering a trap that is not full remains in it. Thus, by inserting $2d(C_l)$ agents we make all traps full.
On the other hand, if an agent passes through all traps and is released to state $X$,
this increases the $\alpha_a$ value by 1 in all non-full traps $a$.
Such an increase can occur at most $m$ times. The remaining additional agents passing
through the line decrease the value of $d(C_l)$.
Therefore, $d(C_l)+m$ additional agents are enough to reduce
the value of $d(C_l)$ to 0.
\end{proof}

In what follows, we consider a scenario where agents can enter the entrance gate of the line during the execution of the protocol. We will prove the following three lemmas.

\begin{lemma}\label{monotone}
%We consider the protocol in a single line starting from configuration $C_l$ and assume agents can enter the line.
Suppose some agents can enter the line during the execution of the protocol, starting from configuration $C_l$. Once $r(C_l)$ tokens with the smallest numbers are handled, the line reaches a configuration $C_l^*\ge \overline{C_l}$, as defined in Lemma \ref{unique}. Between configurations $C_l$ and $C_l^*,$
$s(C_l)$ agents are released from the line.
Moreover, the number of tokens remaining in the line in $C_l^*$ does not exceed the number of agents that entered the line during this process.
\end{lemma}

\begin{proof}
We will demonstrate that each token present in the initial configuration $C_l$ is always handled in the same way, regardless of the interactions generated by the scheduler.
Namely, it is either consistently handled by releasing an agent from the line, or it is consistently eliminated in the same trap $a$.
Let $C_l'$ be a configuration that occurs (not necessarily immediately) after the configuration $C_l.$
%Let $k$ be the smallest token label present in traps $b>a$ in $C_l'$ or $k=k_{\max}+1$ where $k_{max}$ is the maximal token number. 
Let $k_1$ be the largest token number which occurred in traps $b\le a$ in configurations not later than $C_l',$ or was eliminated in trap $a$ before $C_l'$.
%Let us denote by $k_1$ the largest token number that has been moved to or eliminated in trap $a$ until configuration $C_l'$.
Let $k_2$ be the smallest token number that is present in traps $b>a$ in configuration $C_l'.$
If such a token does not exist, then $k_2$ is a number one greater than the maximum token number that has appeared before configuration $C_l'$. Let $k$ be an arbitrary integer, s.t., $k_1\le k<k_2$.
We prove that $\alpha_a(C_l')=\alpha_a^k$ and  $\delta_a(C_l')=\delta_a^k$, for some values
$\alpha_a^k,\delta_a^k$ which depend solely on the initial configuration $C_l$.
We will use time-based induction for the protocol duration.
Let $C_l''$ be the configuration
immediately following $C_l'$.
If no agent moves from trap $a+1$ to $a$ during the transition preceding $C_l''$,
 $\alpha_a(C_l'')=\alpha_a(C_l')=\alpha_a^k$ and  $\delta_a(C_l'')=\delta_a(C_l')=\delta_a^k$.
Otherwise, an agent moves from trap $a+1$ to $a$, and let $k$ be the smallest token label in trap $a+1$.
We can compute that if $\delta_a^{k-1}=0$,
then token $k$ is eliminated and $\alpha_a(C_l'')=\alpha_a^k=\alpha_a^{k-1}$ and $\delta_a(C_l'')=\delta_a^k=1$.
Also, if token $k$ is eliminated in trap $a$, then for traps $b<a$,
$\alpha_b(C_l'')=\alpha_b^k=\alpha_b^{k-1}$ and $\delta_b(C_l'')=\delta_b^k=\delta_b^{k-1}$.
When token $k$ is not eliminated, it is moved to trap $a$.
In this case if $\alpha_a^{k-1}<m$,
then $\alpha_a(C_l'')=\alpha_a^k=\alpha_a^{k-1}+1$ and $\delta_a(C_l'')=\delta_a^k=0$.
When $\alpha_a^{k-1}=m$,
then $\alpha_a(C_l'')=\alpha_a^k=m$ and $\delta_a(C_l'')=\delta_a^k=1$.

We have $\alpha_a(\overline{C_l})=\alpha_a^{r(C_l)}$ and  $\delta_a(\overline{C_l})=\delta_a^{r(C_l)}$ for all $a$.
Since we have assumed that the $r(C_l)$ tokens with the lowest numbers in $C_l^*$ have been processed, then equalities
$\alpha_a(C_l^*)=\alpha_a^{k(a)}$ and $\delta_a(C_l^*)=\delta_a^{k(a)}$ hold for some numbers $k(a)\ge r(C_l)$.
Therefore $C_l^*\ge \overline{C_l}$.
\end{proof}

\begin{lemma}\label{line_throughput}
Let us consider a scenario where a line $l$ starts in a configuration $C_l$. During a time period of $T=3m^2n(\lceil\log n\rceil+1)$, with high probability, the line releases at least $s(C_l)$ agents reaching a configuration $C_l^*\ge \overline{C_l}$.
The number of tokens remaining in the line in $C_l^*$ does not exceed the number of agents which entered the line during execution of the protocol.
\end{lemma}

\begin{proof}
We can partition a time period of length $T=3m^2n(\lceil\log n\rceil+1)$ into $3m$ successive time segments, each with a length of $mn(\lceil\log n\rceil+1)$, corresponding to traps $a=3m,\ldots,2,1$.
We adopt notations from the proof of Lemma \ref{unique}.

According to Lemma \ref{l:surplus}, when the time period corresponding to trap $a$ ends, in total at least $x_{a-1}$ agents reach the gate state of trap $a-1$ from traps with higher indices.
By the end of the time related to trap $1$, at least $x_0=s(C_l)$ agents reach the extra state~$X$.
During this time, all inner states of trap $a$ occupied in $\overline{C_l}$ become occupied.
If also the gate state of trap $a$ is occupied in $\overline{C_l}$,
then at least $\alpha_a(\overline{C_l})+1$ agents are in trap $a$ in configuration $C_l^*$. Therefore $\overline{C_l}\le C_l^*$.
\end{proof}

The last lemma of this subsection provides an estimation for the time required to handle the token with the smallest number.

\begin{lemma}\label{line_throughput1}
Let us consider a scenario where a line starts in the configuration $C_l$. After a time period of length $3mn$, with high probability, 
a token with the smallest number is handled.
\end{lemma}

\begin{proof}
Let us assume that the token with the smallest number is in trap $a$. Within $mn$ time, according to Lemma \ref{l:surplus}, in all traps $b<a$ the gate state becomes nonempty.
Also in time $mn$ trap $a$ releases one agent.
Moving one agent associated with the token from trap $a$ via traps with nonempty gate states to the trap that will eliminate it or to the end of the line reqiures at most $3m$ interactions of pairs of agents in gate states.
It takes time at most $\frac{1}{2}\cdot 3mn$ on average and at most $2mn$ time whp.
This demonstrates that indeed, within $3mn$ time units, this token is handled whp.
\end{proof}

\subsection{Protocol definition}\label{s:prot_def}

For the clarity of presentation, we analyse a ranking protocol defined for $n=3m^3(m+1)$ and even~$m.$
For $3m^3(m+1)<n<3(m+1)^3(m'+2)$, one can arbitrarily scatter $n-3m^3(m+1)$ states by adding up to 2 states to each trap and keep the same asymptotic bounds on time as for $n=3m^3(m+1)$. 
%However, this protocol can be easily modified (by adding a bounded number of states to each trap) to an arbitrary~$n,$ with the same time guarantees.
Thus asymptotically $n=\Theta(m^4)$.
The protocol's rank states are divided into $m^2$ lines of traps.
We mark them with triplets 
$(l,a,b)$, where $l\in[1,m^2]$ is the line number, $a\in [1,3m]$ is a trap number and $b\in [0,m]$ is a trap state.
There is also an extra non-rank state $X$.
Subsets of states having a fixed first component $l$ form lines of traps.
Agents in each line $l$ and trap $a$ are subject to the rule:
\[
(l,a,b)+(l,a,b)\rightarrow (l,a,b)+(l,a,b-1), \mbox{ for } b>0.
\]
Transitions moving agents along the line $l$ are subject to:
\[
(l,a,0)+(l,a,0)\rightarrow (l,a,m)+(l,a-1,0) \mbox{ for } a>1.
\]
Finally, the exit gate in the last trap of line $l$ releases agents to state $X$ according to:
\[
(l,1,0)+(l,1,0)\rightarrow (l,1,m)+X.
\]

\begin{figure}[h]
    \centering
    \includegraphics[scale=0.5]{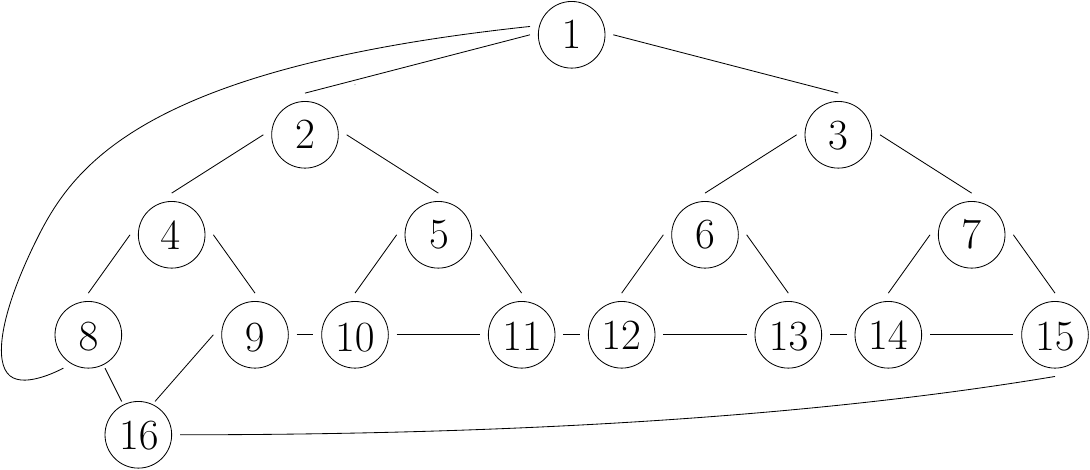}
    \caption{Graph $G$ for $m^2=16$. For example, for $l=1$ we get $l_0=2, l_1=3,$ and $l_2=8$.}
    \label{fig:linestructure}
\end{figure}

Now we define the transition rules for interactions in which an agent with state $X$
participates.
These interactions spread agents more or less evenly across the input gates of all lines.
First we define a cubic graph $G$ with the set of vertices $V=[1,m^2]$ of
diameter $4\lceil\log m\rceil$.
We begin with $G'$ being a balanced binary tree of height
$2\lceil\log m\rceil$
with $m^2+1$ vertices in which every parent has two children.
In this tree there are $m^2/2+1$ leaves and the root has degree 2.
We merge the root with one of the leaves into a single vertex.
Then we add a cycle connecting all remaining leaves obtaining $G$.
Vertices of graph $G$ correspond to lines of traps, see Figure~\ref{fig:linestructure}.
Denote the neighbours of line $l$ in this graph $l_0,l_1,l_2$.
We define the transitions as follows for any $l\in[1,m^2]$.
These transitions define a `routing table' for agents in state $X$.
State $X$ directs agents to line 1:
\[
X+X\rightarrow X+(1,3m,0).
\]
Rank states direct agents to lines designated by the graph $G$:
\[
(l,a,b)+X\rightarrow (l,a,b)+(l_i,3m,0) \mbox{ for } im< a\le (i+1)m \mbox{ and }i\in\{0,1,2\}.
\]
We can note that all states belonging to one trap of a line direct agents to the same line $l_i$
or {\em point to the line} $l_i$.
For this reason we say that each {\em trap points to a line} $l_i$.

Let $C$ be a configuration of all agents in time $t$.
The configuration $C$ truncated to the line $l$ will be denoted by $C_l$.
For each line $l$, we have deficit $d(C_l)$, surplus $s(C_l)$, and excess $r(C_l)$ as in the previous subsection, and $|C_X|$ denotes the number of agents in extra state $X$.
We define for configuration $C$ its surplus $s(C)=|C_X|+\sum_l s(C_l)$, a measure of global flow, and deficit $d(C)=\sum_l d(C_l)$,
which is a measure of the distance of $C$ to the final configuration.
We also define excess $r(C)=|C_X|+\sum_l r(C_l)$ as the total number of tokens (we treat agents in state $X$ as tokens).
Since for each line $s(C_l)\le r(C_l)$, we get $s(C)\le r(C)$.

\begin{lemma}
    For any configuration $C$ we have $s(C)=d(C)$.
\end{lemma}

\begin{proof}
We have
\begin{align*}
    s(C)=&|C_X|+\sum_l s(C_l) =
    |C_X|+\sum_l |C_l|-\sum_l |\overline{C_l}|=
    n-\sum_l |\overline{C_l}|\\=&
    \sum_l 3m(m+1)-\sum_l |\overline{C_l}|=
    \sum_l d(C_l)=d(C).
\end{align*}
\end{proof}

%We say that a trap points to a line $l$ if an agent having state $X$, upon interacting with an agent in this line, transitions into the entrance gate $(l,3m,0)$ of line $l$.
%Let $(l, a, b)$ be a rank state. We say $(l, a, b)$ \emph{points} to a line $l'$ if $l$ and $l'$ are neighbours in $G$.

For each line there are $3m(m+1)$ states of traps pointing to it.
We say that a line is {\em indicated}, if more than $m(m+1)$ states of traps pointing to it are occupied.
The probability of an agent in state $X$ being directed to a given indicated line exceeds~$\frac{m(m+1)}{n}\ge \frac{1}{3m^2}$.
The \emph{flow} $s_{\mathcal{T}}$ in the time interval $\mathcal{T}$ is the number of interactions that move an agent from state $X$ to the entrance gate of any line.

\begin{lemma}\label{r<s}
Let $C$ be a configuration at time $t_0$ and $C'$ be a configuration after a time interval of a length of $3m^2n(\lceil\log n\rceil+1)$ following $C$. 
We have $r(C') \leq s(C)$ whp.   
\end{lemma}

\begin{proof}
    By Lemma \ref{line_throughput} during this time period 
    each line $l$ handles all $r(C_l)$ of its tokens whp, of which
    $r(C_l)-s(C_l)$ tokens are eliminated. % and $s(C_l)$ are handled by releasing an agent to $X$.
    Each time a token is eliminated the value $r(C)$ is reduced by one, and it is not changed in other interactions.
    So $r(C')\le r(C) - \sum_l (r(C_l)-s(C_l))=s(C)$.
\end{proof}

\begin{lemma}\label{throughput}
Let $C$ be a configuration at time $t_0$. Let $\mathcal{T}$ be a time interval starting at $t_0$ with a length of $3m^2n(\lceil\log n\rceil+1)+O(\log n)$. We have $s_{\mathcal{T}} \geq s(C)$ whp.
\end{lemma}

\begin{proof}
    By Lemma \ref{line_throughput}, each line of traps $l$ releases to state $X$ in time whp $3m^2n(\lceil\log n\rceil+1)$ at least $s(C_l)$ agents.
    Adding the $|C_X|$ agents already in state $X$, the total number sums up to $s(C)$ agents.
    The probability that a given agent leaves state $X$ in a given interaction is at least $1/n$.
    The probability that this agent does not leave in time $c \ln n$ is
    \[ \left(1-1/n\right)^{cn\ln n}\le e^{-c\ln n}=n^{-c}.\]
    Thus, all of these agents leave state $X$ after time $O(\log n)$ whp.
\end{proof}

\begin{lemma}\label{throughput1}
    Let $\mathcal{T}$ be a time interval of length $c m^2 n \log n$, for some constant $c$ that can be derived from the proof. 
    Configurations at the beginning and the end of period $\mathcal{T}$ are $C$ and $C'$, respectively.
    Assume that $r(C)<m^3$ %, $0.9 r(C)<r(C')$, 
    and $m^2<r(C')$.
    Then the total flow of agents $s_{\mathcal{T}}$ during the period $\mathcal{T}$ is at least $m^3$ whp.
\end{lemma}

\begin{proof}
    Assume that on the contrary $s_{\mathcal{T}}<m^3$.
    Let us divide the period $\mathcal{T}$ into two parts: the initial part $\mathcal{T}_0,$ with a length $3m^2n(\lceil\log n\rceil+1)$, and the remaining part denoted as $\mathcal{T}'$. We~have $s_{\mathcal{T}_0}\le s_{\mathcal{T}}<m^3$.
    Thus,
    less than $m^3$ tokens are added to all lines during the period~${\mathcal T}_0$. According to Lemma~\ref{line_throughput}, after $\mathcal{T}_0$, all tokens initially present in these lines are handled whp. Consequently, the tokens present in the lines after the period $\mathcal{T}_0$ are exclusively those that entered the lines during 
    $\mathcal{T}_0$.

    For the purpose of protocol analysis, we distinguish two types of agents: {\em regular} and {\em irregular}.
    In particular, directly before each interaction, we distinguish as many as possible but not more than $m+1$ agents in any trap as regular.
    The remaining agents, including those in state $X$, are irregular.
    Furthermore, it's important to note that the number of tokens $r(C)$ is always greater than or equal to the number of irregular agents.
  %  Note also that the number of tokens $r(C)$ is not smaller than
  %  the number of irregular agents.
%

An agent in state $X$ enters a line by initiating interaction with another agent, who can be regular or irregular. Upon entering the line, it transforms into a token on that line. 
Tokens generated from interactions with regular agents increase the count of {\em black} tokens on the line, whereas those from interactions with irregular agents increase the count of {\em red} tokens. Black tokens consistently occupy traps
on the line with the lowest numbers.
The number of irregular agents does not exceed $r(C)<m^3$. Thus, the probability of increasing the count of red tokens when an agent enters the line does not exceed $\frac{m^3}{n-m^3}<\frac{1}{3m}$. This way, Lemma~\ref{line_throughput1} guarantees that every $3mn$ time units one of these black tokens will be handled.

    %Let $r(C')>m^2$ be the total number of tokens at the end of the period $\mathcal{T}$.
    Since in period $\mathcal{T}$ 
    the flow $s_{\mathcal{T}}<m^3$, the expected total number of red tokens created during the period $\mathcal{T}$ does not exceed $m^3\cdot\frac{1}{3m}$. By the Chernoff bound, this value does not exceed $m^2/2$~whp.
    Note, that the total number of tokens in any chunk is not smaller than $r(C')>m^2$, so the number of black tokens exceeds $m^2/2$.

    %So, since $r(C')\ge m^2$, then in any configuration present during the period $\mathcal{T}$, we have at least $m^2/2\le r(C)-m^2/2$ black tokens.
    Let the period $\mathcal{T}'$ be the sum of consecutive chunks $\mathcal{T}_1,\mathcal{T}_2,\ldots,\mathcal{T}_j$.
    The duration of each chunk $\mathcal{T}_i$ is $3c's_{\mathcal{T}_{i-1}}n\log n/m$, where $c',j$ will be defined later in the proof.

    We show by induction that all black tokens in the lines are handled during the chunk $\mathcal{T}_i$ with high probability.
    In the previous chunk $\mathcal{T}_{i-1}$, the flow is $s_{\mathcal{T}_{i-1}}\ge m^2/2$, because all black tokens are handled. On average, at most $\frac{3m(m+1)s_{\mathcal{T}_{i-1}}}{n-m^3}$ black tokens will enter each line. By the Chernoff bound, there exists a $c'$ such that with high probability at most $c's_{\mathcal{T}_{i-1}}\log n/m^2$ black tokens will enter each line. According to Lemma~\ref{line_throughput1}, the time $3c's_{\mathcal{T}_{i-1}}n\log n/m$ ensures that all black tokens present in the lines at the beginning of chunk $\mathcal{T}_i$ will be handled with high probability.

    The number of chunks is the smallest  $j$ such that
    \(
    \sum_{i=0}^{j} s_{\mathcal{T}_{i}}\ge m^3.
    \)
    These chunks take time $3m^2n(\lceil\log n\rceil+1)+3c'm^2n\log n=O(m^2n\log n)$, and generate a total flow of at least $\sum_{i=0}^{j} s_{\mathcal{T}_{i}}\ge m^3$ guaranteed by black tokens alone. This contradicts the assumption~$s_{\mathcal{T}}<m^3.$
\end{proof}

\begin{lemma}\label{timetosqr}
    After initial time $O(m^3n\log^2 n)$ we have $r(C)<m^2$.
\end{lemma}

\begin{proof}
    Let us divide the time until $s(C)<m^3$ into intervals $\mathcal{T}$ of length $3m^2n(\lceil\log n\rceil+1)+O(\log n)$ as in Lemma \ref{throughput}. 
    By Lemma \ref{throughput} in each interval at the beginning of which $s(C)\ge m^3$, we also have $s_{\mathcal{T}}\ge m^3$.
    By Lemma \ref{r<s}, after the first interval for which $s(C)<m^3$ 
    we also have $r(C')<m^3.$
    
    The remaining time we divide into intervals $\mathcal{T}$ of length $cm^2n\log n$ as in Lemma \ref{throughput1}.     
    %During this time there can be at most $O(\log m)$ static intervals for which $0.9r(C)>r(C')$.
    %During this time there can be at most $O(\log m)$ intervals for which $0.9r(C)>r(C')$, when Lemma \ref{throughput1} does not apply.
    By Lemma \ref{throughput1} in any remaining interval for which $r(C')\ge m^2$ we have $s_{\mathcal{T}}\ge m^3$.
    
    After the initial $4m$ intervals in which $s_{\mathcal{T}}\ge m^3$ and additional time $3m^2n(\lceil\log n\rceil+1)$ some line $l$ becomes full whp. 
    Indeed these intervals guarantee the total flow at least $4m^4$, so 
    the line receiving the most tokens receives at least $4m^2$ of them.
   After an additional time of $3m^2n(\lceil\log n\rceil+1)$, this line becomes full whp by Lemmas \ref{reduce_d} and \ref{line_throughput}.
    
    In any additional $10m$ intervals in which $s_{\mathcal{T}}\ge m^3$ there is a total flow at least $10m^4$.
    Any line neighbouring to full lines in graph $G$ is indicated.
    The expected number of tokens entering any neighbouring line is at least $\frac{10}{3}m^2$.
    By the Chernoff bound, $3m(m+1)+m$ tokens enter any neighbouring line in this time whp.
    In additional time $3m^2n(\lceil\log n\rceil+1)$ these lines become full whp by Lemmas \ref{reduce_d} and \ref{line_throughput}.

    Thus, after $4\log n$ time periods from the previous paragraph all lines become full or $s(C)<m^2$ whp.
    In both cases $s(C)<m^2$.
    During an additional time of $3m^2n(\lceil\log n\rceil+1)$, 
    by Lemma~\ref{r<s}, we have $r(C)<m^2$.
\end{proof}

\begin{lemma}\label{all_indicated}
    If $r(C)<m^2$, then all lines are indicated.
\end{lemma}

\begin{proof}
    If a line is not indicated, then traps indicating this line are lacking more than $2m(m+1)$ agents.
    These agents cause excesses in other traps yielding more than $2m(m+1)$ tokens.
    This contradicts the assumption that the total number of tokens is $r(C)<m^2$.
\end{proof}

\begin{lemma}\label{throughput2}
    Let $\mathcal{T}$ be a time interval of length $c m^2 n \log n,$ for some constant $c$ that can be derived from the proof. 
    Configurations at the beginning and the end of the period $\mathcal{T}$ are $C$ and $C'$ respectively.
    Assume that $r(C)<m^2$ %, $0.9 r(C)<r(C')$, 
    and $r(C') = k < m^2$.
    Then the total flow of agents $s_{\mathcal{T}}$ during the period $\mathcal{T}$ is at least $k\cdot m$ whp.
\end{lemma}

\begin{proof}
    Let us divide the period $\mathcal{T}$ into two parts: the initial part $\mathcal{T}_0$ with a length $3m^2n(\lceil\log n\rceil+1)$, and the remaining part denoted as $\mathcal{T}'$. In the scenario where the flow during
the initial period $\mathcal{T}_0$ satisfies $s_{\mathcal{T}_0}\ge m^3$,
    a total flow also satisfies $s_{\mathcal T}\ge m^3\ge km$. Conversely,
    in the remaining case, one gets $s_{\mathcal{T}_0}<m^3$.
    %Let us assume, then, that $s_{\mathcal{T}_0}<m^3$. 
    Thus,
    less than $m^3$ tokens are added to all lines during the period ${\mathcal T}_0$. According to Lemma~\ref{line_throughput}, after $\mathcal{T}_0$, all tokens initially present in these lines are handled whp. Consequently, the tokens present in the lines after the period $\mathcal{T}_0$ are exclusively those that entered the lines during 
    $\mathcal{T}_0$.

    Let us divide the period $\mathcal{T}'$ into consecutive chunks $\mathcal{T}_1,\mathcal{T}_2,\ldots,\mathcal{T}_j$. The duration of each chunk $\mathcal{T}_i$ is $3c's_{\mathcal{T}_{i-1}}n\log n/m$, where $c'$ will be defined later in the proof.
    We show by induction that all tokens in the lines are handled during the chunk $\mathcal{T}_i$ whp. Indeed, in the preceding chunk $\mathcal{T}_{i-1}$, the flow is $s_{\mathcal{T}_{i-1}}\ge r(C')=k$, because all tokens are handled at least once whp.
    For a given line, the probability that a token will be inserted at it's entrance gate is not greater than $\frac{3m(m+1)+m^2}{n-1}<\frac{4}{3m^2}$.
    On average, at most $\frac{4s_{\mathcal{T}_{i-1}}}{3m^2}$ tokens enter a given line.
    In the case when $s_{\mathcal{T}_{i-1}}\ge m^2$, by the Chernoff bound, 
    there exists a constant $c'$ such that whp at most $c's_{\mathcal{T}_{i-1}}\log n/m^2$ tokens will enter each line.
    According to Lemma \ref{line_throughput1}, the time $3c's_{\mathcal{T}_{i-1}}n\log n/m<3c's_{\mathcal{T}_{i-1}}mn\log n/k$ guarantees that all tokens present in the lines at the beginning of chunk $\mathcal{T}_i$ will be handled whp.
    In the case when $s_{\mathcal{T}_{i-1}}< m^2$, by the Chernoff bound, 
    there exists a constant $c'$ such that whp at most $c'\log n$ tokens will enter each line.    
    Finally, according to Lemma \ref{throughput1}, the time $3c'mn\log n<3c's_{\mathcal{T}_{i-1}}mn\log n/k$ guarantees that all tokens present in the lines at the beginning of chunk $\mathcal{T}_i$ will be handled whp.

    The number of chunks is the smallest number $j$ such that
    \(
    \sum_{i=0}^{j} s_{\mathcal{T}_{i}}\ge mk.
    \)
    These chunks take time not greater than $3m^2n(\lceil\log n\rceil+1)+3c'm \cdot mk\log n/k=O(m^2n\log n)$, and generate a total flow of at least $\sum_{i=0}^{j} s_{\mathcal{T}_{i}}\ge mk$.
\end{proof}

\begin{lemma}\label{timetolog}
    Assume $r(C)=2k<m^2$, for $k>\ln n$. After time $O(m^3n\log n)$ we get $r(C)<k$~whp.
\end{lemma}

\begin{proof}
   Let us divide the time until $r(C)<k$, into intervals $\mathcal{T}$ of length $c m^2 n \log n,$ where $c$ is specified by Lemma~\ref{throughput2}.
    %There can be a constant number of static intervals for which $0.9r(C)>r(C')$.
    By Lemma \ref{throughput2} in any such interval $s_{\mathcal{T}}\ge km$ and by Lemma \ref{all_indicated} all lines are indicated.

    For any constant $c''$ in $4c''m$ first intervals there is a total flow at least $4c'' m^2 k$ whp by Lemma~\ref{throughput2}. 
    By the Chernoff bound, for any line at least $2k$ agents enter this line in this period.
    For each line $l,$ lemmas \ref{reduce_d} and \ref{line_throughput} guarantee reduction of $d(C_l)$ by at least half.
    Thus, $d(C)=s(C)$ is also reduced from $2k$ to less than $k$.
    After an additional time of $3m^2n(\lceil\log n\rceil+1)$, 
    by Lemma \ref{r<s}, we have $r(C)<k$ whp.
\end{proof}

\begin{lemma}\label{timetozero}
    Assume $r(C)\le 2\ln n$. After time $O(m^3n\log^2 n)$ full stabilisation occurs whp.
\end{lemma}

\begin{proof}
    Let us divide the time when $r(C)\ge 1$ into intervals $\mathcal{T}$ of length $c m^2 n \log n,$ where $c$ is specified by Lemma~\ref{throughput2}. 
    %There can be $O(\log\log n)$ of static intervals for which $0.9r(C)>r(C')$.
    Also by this lemma, in any such an interval $s_{\mathcal{T}}\ge m$ whp and by Lemma \ref{all_indicated} all lines are indicated.

    In the first $8c''m$ intervals, where $c''$ is a constant, there is a total flow exceeding $8c'' m^2 \log n$ whp.
    %For a constant $c''$, in $4c''m$ first intervals there is a total flow exceeding $4c'' m^2 \log n$ whp. 
    By the Chernoff bound, there is $c''$ such that in this period at least $4\ln n$ agents enter each line whp.
    This guarantees that all lines become full by Lemmas \ref{reduce_d} and \ref{line_throughput}.
\end{proof}

\begin{theorem}
    By solving ranking with one extra state,
    we obtain a silent, self-stabilising
leader election protocol with a running time $O(n^{7/4}\log^2 n)$~whp.
\end{theorem}

%\begin{theorem}
%    The ranking protocol with one extra state stabilises silently in time $O(n^{7/4}\log^2 n)$~whp.
%\end{theorem}

\begin{proof}
    After time $O(m^3n\log^2 n)$ the protocol reaches a configuration $C$ for which $r(C)<m^2$ by Lemma
    \ref{timetosqr}.
    After each subsequent period of length $cm^3n\log n$ from Lemma \ref{timetolog}, the value of $r(C)$ decreases by at least a factor of two.
    Therefore, within time $O(m^3n\log^2 n)$, the value $r(C)\le 2\ln n$.
    Finally, the protocol stabilises after an additional time of $O(m^3n\log^2 n)$ by Lemma \ref{timetozero}.
    Thus, the total time until stabilisation is $O(m^3n\log^2 n)=O(n^{7/4}\log^2 n)$.
\end{proof}

\iffalse
\begin{lemma}
    Jeśli S tokenów rozdzielimy losowo na M linii, to każda linia otrzyma whp co najwyżej $(1+2\eta)\mu$ tokenów
    gdy $\mu>\ln n$ i co najwyżej $\mu+2\eta\ln n$ tokenów, gdy $\mu\le\ln n$.
\end{lemma}

\begin{proof}
    Ustalona linia otrzymuje średnio $\mu=S/M$ tokenów.
    Dla jakiego $\delta$ otrzyma co najwyżej $(1+\delta)\mu$ tokenów whp?
    \[\Pr(X>(1+\delta)\mu))<e^{-\delta^2\mu/(2+\delta)}<n^{-\eta} \]
    Przypadek 1: $\mu>\ln n$. Wtedy dla $\delta=2\eta,\eta>1$ mamy
    \[e^{-\delta^2\mu/(2+\delta)}<e^{-4\eta^2\ln n/(2+2\eta)}<e^{-\eta\ln n}=n^{-\eta}\]

    Przypadek 2: $\mu\le\ln n$. Wtedy dla $\delta=2\ln n\eta/\mu,\eta>1$ mamy
    \[e^{-\delta^2\mu/(2+\delta)}<e^{-4\ln^2 n\eta^2/(2\mu+2\ln n\eta)}<e^{-\eta\ln n}=n^{-\eta}\]
\end{proof}
\fi

\section{Ranking protocol with O(log n) extra states}\label{s:logn}

In this section, we introduce a self-stabilising ranking protocol that utilises $O(\log n)$ extra states, stabilising in time $O(n\log n)$ whp. This represents an exponential improvement in extra states utilisation compared to~\cite{BCC+21}, which relies on $\Omega(n)$ extra states. 
The proposed solution utilises a novel concept of {\em perfectly balanced} binary trees, defined for any integer size $n$. This tree is constructed recursively, starting from its root. %referring to rank state $0$.
%In the recursive step of tree construction, we start from its root denoting rank $p$.
%The task is to construct a subtree of size $k$.
If the size $k$ (starting with $k=n$) of the rooted tree is odd, i.e., $k=2l+1,$ the root is a {\em branching node} that spawns two children, serving as the roots of two identical (also perfectly balanced) subtrees of size $l.$ For even $k$, the root is a {\em non-branching node}, acting as a parent to a leaf or a branching node that serves as the root of a subtree of size $k-1.$
We also define the level of a node as its distance $d$ from the root. Based on the full symmetry of the recursive construction provided above, perfectly balanced binary trees exhibit several properties, including: (1) all nodes (and subtrees rooted) at the same level are uniform, i.e., they are either branching nodes with two children (the roots of two subtrees of the same structure and size), have one child, or are leaves in the tree, and (2) the height $h$ of the tree satisfies $h \leq 2\log n$.

\begin{figure}[h]
    \centering
    \includegraphics[scale=0.5]{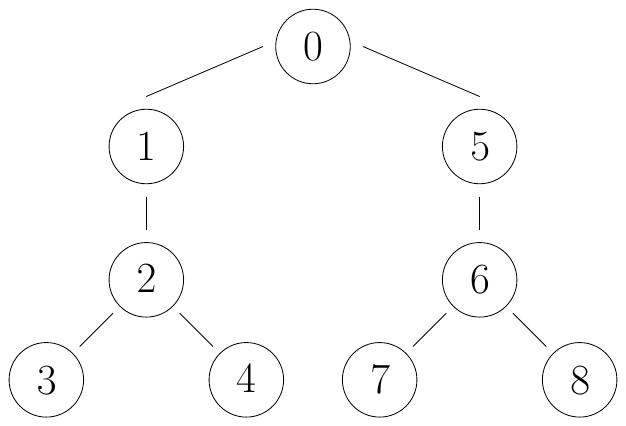}
    \caption{State distribution in perfectly balanced binary tree for $n=9$.}
    \label{fig:pbt}
\end{figure}

In our protocol, the nodes of a perfectly balanced binary tree are used to span all $n$ rank states, % of the ranking protocol.  
where the node with the {\em pre-order} number $0\le p\le n-1$ in this tree represents state $p$. 
In this order, the root represents state $0,$ a lone child of a parent representing state $p$ corresponds to state $p+1,$ and children of a branching node representing state $p$ correspond to states $p+1$ (left child) and $p+l+1$ (right child), where $l$ is the size of each of the two identical trees rooted in nodes representing states $p+1$ and $p+l+1$, see Figure~\ref{fig:pbt}.
We will refer to this tree as the {\em tree of ranks}.

We define a rank state $p$ as {\em overloaded} if it is occupied by at least two agents. 
Rule ${\mathcal R}_1$ of the transition function manages interactions within overloaded states as follows.
%defined on the ranking states organised in perfectly balanced binary tree of size $n$ is as follows.
%
\begin{comment}
\[
{\mathcal R}_1: (p,p)\rightarrow (p,p+1) \mbox{, if } p+1 \mbox{ is a lone child, and}
\]
%
\[
{\mathcal R}2: (p,p)\rightarrow (p+1,p+k) \mbox{, if } p \mbox{ is a branching node}.
\]
\end{comment}
%
\[
        {\mathcal R}_1: p+p\rightarrow  \begin{cases}
        p+(p+1), & \mbox{if } p \mbox{ is a non-branching node, and}\\
        (p+1)+(p+l+1), & \mbox{if } p \mbox{ is a branching node}.
        \end{cases}
\]

Namely, when two agents interact while occupying a non-branching state $p$, only the responder adopts state $p+1$. 
In contrast, if they interact in a branching state $p$, both agents vacate $p$. 
The initiator adopts state $p+1$, while the responder adopts state $p+l+1$.

We first observe in Lemma~\ref{l:tree}, that rule ${\mathcal R}_1$ admits $O(n\log n)$ time perfect dispersion of agents in the tree of ranks. %, if possible. 
%which aligns with the structure of $n$-state automaton with transition rule ${\mathcal R}_1.$
In Section~\ref{s:nlogn}, we show that if needed, the protocol based on rules ${\mathcal R}_2, {\mathcal R}_3, {\mathcal R}_4,$ and ${\mathcal R}_5$ guarantees relocation of all agents to the buffer line formed of $O(\log n)$ extra states $X_1,\dots,X_{2k}$ in time $O(\log n),$ before allowing agents to (re)enter the tree of ranks, see Lemma~\ref{l:upper-line}, and conclude with perfect agent dispersion among rank states.
These observations yield the following theorem, with its proof deferred to the end of Section~\ref{s:logn}.

\begin{theorem}\label{t:logn}
    By employing the ranking protocol based on rules ${\mathcal R}_1, \dots,{\mathcal R}_5$ and $O(\log n)$ extra states, we obtain a silent, self-stabilising leader election protocol with a running time $O(n\log n)$~whp.
\end{theorem}

%\begin{theorem}\label{addlog}
%    The ranking protocol based on rules ${\mathcal R}_1, \dots,{\mathcal R}_5$ utilises $n$ ranks states and $x=O(\log n)$ extra states stabilises silently in time $O(n\log n)$ whp.
%\end{theorem}

Before we present and analyse the complete self-stabilising ranking protocol (see Section~\ref{s:nlogn}), we focus on the progress of agents along {\em root-to-leaf} paths. These paths connect the root of the tree of ranks with its leaves.
In particular, we say that a root-to-leaf path $P$ is {\em balanced} if the protocol, based on rule ${\mathcal R}_1$, stabilises silently in a configuration where each node (rank) on this path is occupied by a single agent. Conversely, we say that such a path is {\em overloaded} if this protocol leads to a configuration where the last node on this path (a leaf in the tree of ranks) becomes overloaded.

\begin{lemma}\label{l:root}
In the initial configuration, if all rank states except the root $0$ are unoccupied, and all $n$ agents are at or will eventually arrive at the root, then rule ${\mathcal R}_1$ ensures that all rank states will eventually become occupied.
\end{lemma}
\begin{proof}
According to the definition of perfectly balanced trees, if $n$ is even the root is a non-branching node.
In this case, agents are moved to (lone child) state $1$ until a single agent is left in state $0.$  
If $n$ is odd, the same amounts of $\frac{n-1}{2}$ agents are moved to states $1$ (left child) and 
$\frac{n-1}{2}+1$ (right child) until a single agent is left in state $0.$   
This process is repeated recursively at each internal node of the tree of ranks. 
Thus due to perfect distribution, in the final configuration all (including leaf) 
states are occupied by single agents, i.e., the ranking process is successfully completed.
\end{proof}

%\subsection{Fast non-self-stabilising ranking with no extra states}\label{s:non}
\subsection{Progress along root-to-leaf paths}

We observe first that if all agents occupy rank states, 
either all root-to-leaf paths in the tree of ranks are balanced 
(i.e., the protocol based on rule ${\mathcal R}_1$ will successfully rank all agents) or
at least one such path is overloaded (i.e., further actions are needed to rank the population). 

We prove the following lemma.

\begin{lemma}\label{l:tree}
    Given either balanced or overloaded root-to-leaf path. 
    The ranking protocol based on rule ${\mathcal R}_1$ either stabilises evenly or
    overloads the leaf on this path in time $O(n\log n)$ whp. 
\end{lemma}
\begin{proof}
    Consider an arbitrary {\em root-to-leaf} path in the tree of ranks.
    Assume first that this path is balanced, i.e., the target number of agents on this path 
    in the final configuration $\overline{C}$ is $h+1$,
    where $h$ is the height of the tree of ranks.
    Let $h+1=h_b+h_n$ refers to  
    $h_b$ branching and $h_n$ non-branching nodes, where we count the leave as a branching node. 
    Consider a configuration $C$ with $k$ agents currently occupying nodes on this path,
    where the split $k=k_b+k_n$ refers to agents residing in nodes counted as branching and non-branching, respectively.
    %Let in $C$ the number of occupied nodes on this path is
    %$l\le h$ where $l_b\le h_b$ of them are branching
    %and $l_n\le h_n$ are non-branching.
    We define a potential function $F(C) = k_b+\frac{3}{2}k_n-h_b-\frac{3}{2}h_n$, and
    let $C$ transition to $C'$ after one interaction. 
    We show that the following inequality holds:
    \[ E[F(C')] < \left(1 - \frac{1}{n^2}\right) F(C). \]

After delegating exactly one agent to every occupied node on this path, which requires at most $h_b+h_n$ agents, 
each of the remaining at least $k_b+k_n-h_b-h_n$ agents engages in a move (down the tree) related interaction 
with the probability greater than $2/n^2.$ 
Each such move decreases the value of $F(C)$ by at least $1/2$.
Therefore, the expected decrease of $F(C)$ during single interaction is at least 
\[\left((k_b-h_b)+\frac{3}{2}(k_n-h_n)\right)\cdot\frac{1}{2} \cdot \frac{2}{n^2}= \left(k_b+\frac{3}{2}k_n-h_b-\frac{3}{2}h_n\right)\cdot \frac{1}{n^2} = F(C) \cdot \frac{1}{n^2}.\]

The value of the potential function in the initial configuration is $F(C_0)< \frac{3}{2}n$. 
The worst case is the deployment of all $n$ agents in non-branching states.
%,as all agents are in the rank state $0$.
If $C_{T}$ is the configuration after $T$ interactions, then
\[
E[F(C_{T})]\le (1-1/n^2)^{T}F(C_0)<\frac{3}{2}e^{-T/n^2}n. 
\]
We have $F(\overline{C})=0$ and for $C_T\neq \overline{C}$ we have $F(C_T)\ge 1/2$.
Using Markov's inequality we get
\[\Pr(C_T\neq \overline{C})=\Pr(F(C_T)\ge 1/2)\le 
    \frac{E[F(C_T)]}{1/2}<3e^{-T/n^2}n.     
\]
We get $3e^{-T/n^2}n\le n^{-\eta}$ for $T\ge n^2((\eta+1)\ln n+\ln 3)$.

Assume now that the considered root-to-leaf path is overloaded. 
Let us extend this path by adding one extra node. 
If the new path becomes balanced, we can adopt the argument from above, 
concluding that the extra node will be populated (indicating that the path is overloaded) in time $O(n\log n)$ whp. 
The availability of extra agents on this path only accelerates the process (increases transfer probabilities) during which the extra node gets populated.
Note that such extra node will be utilised in the ranking algorithm, presented in the next section, 
to trigger the reset signal for the entire system.
\end{proof}

\subsection{Fast self-stabilising ranking with O(log n) extra states}\label{s:nlogn}
%
%Thus we need to show that if in the initial configuration there are some overloaded paths 

Recall that the rank states are in the range $0, \dots, n-1$. We partition the extra states into two groups: $X_1, \dots, X_k$ and $X_{k+1}, \dots, X_{2k}$, which we will refer to as the {\em red} and {\em green} groups, respectively, where $x=2k=O(\log n)$. 
The following rule generates the {\em reset signal} by changing the state of two agents sharing a leaf rank to the red state $X_1$, with the intention of resetting all already ranked agents.
\[
{\mathcal R}_2: l+l\rightarrow X_1+X_1 \mbox{, where } l \mbox{ is a leaf in the tree of ranks.} 
\]
The extra states form a line starting in state $X_1$ and ending in state $X_{2k}$. The exchange of agents' states along this line is governed by the following rules:
\[
{\mathcal R}_3: X_i+X_j\rightarrow X_{i+1}+X_{i+1} \mbox{, where } i\le j \mbox{ and } i<2k. 
\]
The following rule mandates interactions of agents in extra states with those in rank states $j\in\{0,\dots,n-1\}$.
\[
    {\mathcal R}_4: X_i+j\rightarrow  \begin{cases}
        X_{1}+X_{1}, & \mbox{if } i\le k,\\
        0+j, & \mbox{if } i>k.
        \end{cases}
\]
The first case of the rule allows agents with red extra states to unload (reset) 
the tree of ranks and to propagate the reset signal. 
The second case relocates agents with green extra states to the root of the tree of ranks.
Also rule ${\mathcal R}_5$ is responsible for such transfer.
\[
{\mathcal R}_5: X_{2k}+X_{2k}\rightarrow 0+0. 
\]
%

%The ranking protocol presented here utilises $O(\log n)$ extra states.
%We show that this protocol is able (if needed) to move temporarily all agents 
%to the extra states, and subsequently to the root of empty tree of ranks.
Before we analyse performance of the ranking protocol based on rules ${\mathcal R}_1$ through ${\mathcal R}_5,$
we prove one additional lemma governing distribution of agents along the line of $2k$ extra states.
We assume that constant $k$ is large enough to accommodate the needs of the following lemma (i.e., \ $k\ge d'$).

\begin{lemma}\label{l:upper-line}
    There exists $d'>0$, s.t., for any $d\ge 0$ there exists $c>0,$
    for which after at most $c\log n$ time since state $X_1$ (reset signal) arrival, all agents are in line states $X_i$ with indices $d\log n< i\le (d+d')\log n$ whp.
\end{lemma}

\begin{proof}
    Let $t'=c'\log n$ be the time needed to complete epidemic whp amongst $n$ agents.
    Let us consider an arbitrary pair of agents $a$ and $b$, %at the very start of this epidemic.
    where during the epidemic, a {\em chain of interactions} $a+a', a'+a'', 
    \dots
    a^{(l-1)}+b,$ begins with agent $a$, ultimately reaching agent $b$.
    We expand this chain to the {\em extended} chain $\mathcal{C}(a,b)$ which includes also all interactions from the beginning to the end of the epidemic, involving:
    $*+a$ before interaction $a+a';$
    $*+a'$ after and before interactions $a+a'$ and $a'+a'',$ respectively;
    etc, finishing with all interactions $*+b$ occurring after interaction $a^{(l-1)}+b.$
    The expected length of this extended chain is $2c'\log n$.
    By Chernoff inequality there is $d'>0$, s.t.,\ this length is
    smaller than $d'\log n$ whp.

    In the initial period of time $c'\log n$, since an agent $a$ with state $X_1$ occurs, any agent $b$ belongs to some extended chain $\mathcal{C}(a,b)$ of length less than $d'\log n$. 
    Thus at time $c'\log n$ agent $b$ is in red state $X_i$, where $i<d'\log n$.
    The argument showing that after time $t<c\log n$, all agents are in states $X_i$, for $i>d\log n$ (for any $d$), can be found in~\cite{AlistarhDKSU17}. Note that if $t\le c'\log n$, the thesis holds for $t^*=c'\log n$.
    Otherwise, let $a$ be an agent in state $X_i$, for $i\le d\log n$, at time $t'\le t-c'\log n$. We remind that every agent $b$ has an extended chain $\mathcal{C}(a,b)$ of length
    smaller than $d'\log n$ from agent $a$. Thus, at any time up to $t$ any agent $b$ is in state
    $X_i$ for $i<(d+d')\log n$ whp.
%
    \iffalse
    If $t'$ is the time in which
    For a configuration $C$, let $Z$ be the set of agents which have already had direct or indirect 
    contact with the signal initiated by the agent that first adopted state $X_1$.
    We define the potential function $P(C)=\sum_{a\in Z} 2^{i(a)}$,
    where $i(a)$ refers to the index of state $X_i$ that agent $a$ is currently in.
    Any agent $a\in Z$ can at most double its contribution during an interaction. 
    In addition, it can pass this contribution to an agent introduced to $Z$ during this interaction. 
    Therefore, if $C'$ is a configuration following
    $C$ directly, we get
    \[E[P(C')]\le (1+8/n)P(C).\]
    In the initial configuration, $P(C_0)=1$.
    Thus for configuration $C_{cn\log n}$ reached after 
    $cn\log n$ interactions we get:
    \[E[P(C_{cn\log n})]\le (1+8/n)^{cn\log n}P(C_0)\le
    e^{8c\log n}=n^{8c/\ln 2}. 
    \]
    Using Markov's inequality
    \[\Pr(\exists a:i(a)>(8c/\ln 2+\eta)\log n)\le
        \Pr(P(C_{cn\log n})>n^{8c/\ln 2+\eta})\le 
        \frac{E[P(C_{cn\log n})]}{n^{8c/\ln 2+\eta}}\le n^{-\eta}.     
    \]
    Therefore, the thesis holds for $d_u=8c/\ln 2+\eta$.
    \fi
\end{proof}

%We conclude with the final theorem.
%
%\begin{theorem}\label{addlog}
%    The ranking protocol based on rules ${\mathcal R}_1, \dots,{\mathcal R}_5$ utilises $n$ ranks states and $x=O(\log n)$ extra states and stabilises in time $O(n\log n)$ whp.
%\end{theorem}
%

We conclude this section with the proof of the main theorem.

\begin{proof}[Proof of Theorem \ref{t:logn}]
The {\em balanced configuration} is such that if we convert all extra states to the root state $0$ then a protocol starting from this modified configuration will lead to the final silent configuration in which all rank states are occupied.
Consider a {\em modified protocol} in which all states on the line are treated as green and the transition $X_i+j\rightarrow 0+j$ is always performed in rule $\mathcal{R}_4$.
This modified protocol leads to the first overload of a leaf in time $O(n \log n)$ whp when we do not start from a balanced configuration
by Lemmas \ref{l:tree} and \ref{l:upper-line}.
Due to the same lemmas, if we start from a balanced configuration, we obtain the final silent configuration in $O(n \log n)$ time whp.

We can now proceed to analyse the actual protocol.
We consider two cases: starting from a balanced configuration and from a non-balanced configuration.
%We consider two cases: starting from a balanced configuration and starting from a configuration different from the balanced configuration.
If we start from a balanced configuration, two scenarios are possible. 
In the first scenario, there is no interaction between an agent in a red state and an agent in a tree state for time $O(n\log n)$.
The protocol computations in this scenario are identical to those in the modified protocol. Therefore, in this scenario, the probability of not reaching the silent configuration in time $O(n\log n)$ is negligible (less than $n^{-\eta}$).
In the second scenario, an agent in a red state has an interaction with an agent in a tree state in time $O(n\log n)$.
After this happens, all agents in the population receive red states whp through an epidemic in time $O(\log n)$ by Lemma \ref{l:upper-line}.
Then, in time $O(\log n)$, they transition into green states.
The configuration then becomes silent after time $O(n\log n)$ whp, so the probability of failure is also negligible in this scenario.

In the second case, when we start from a configuration different from the balanced configuration, we also consider two scenarios. 
In the first scenario, an agent in a red state from the starting configuration after interacting with an agent in a tree state transitions to state $X_1$ in time $O(n \log n)$. 
The analysis of this scenario is the same as in the second scenario of case one. 
In total time $O(n \log n)$, the configuration becomes silent whp,
so the probability of not reaching the final silent configuration is negligible.
In the second scenario, there is no interaction between an agent which was in a red state from the beginning and an agent in the tree state.
Then the protocol computation is identical to the modified protocol computation and in time $O(n \log n)$ there is an overload of some leaf in the tree.
The analysis of the remaining part of this scenario is again identical to the second scenario of the first case.
In total time $O(n \log n)$, the configuration becomes silent whp,
so the probability of not reaching the final silent configuration is negligible.
\end{proof}

\section{Conclusion}\label{s:conclusion}

We introduced and conducted detailed analysis on multiple sub-quadratic self-stabilising ranking and in turn leader election protocols. The improved efficiency of these protocols hinges on either near-stabilised initial configurations or the utilisation of a small number of extra states. Since ranking guarantees leader election in self-stabilising protocols, our findings also provide solutions for leader election in this model.

The main remaining open problem in this area pertains to the existence of a ranking protocol with a time complexity of $o(n^2)$  based solely on the use of $n$ rank states. It's worth noting that also no efficient (polylogarithmic) $n$ state population protocol is known for solving ranking when all agents start in a uniform state. 
Thus, further understanding of the ranking and naming problems is needed. %in standard population protocols. 
And in particular, how the stabilisation time in both problems changes when the number of extra states decreases.

\section{Chernoff Bound}\label{s:app}

\begin{corollary}[Utilisation of the Chernoff Bound]\label{coro}
Let us randomly distribute $S$ tokens among $M$ lines and denote $\mu=S/M$. Each line receives whp at most $(1+2\eta)\mu$ tokens when $\mu>\ln n$, and at most $\mu+2\eta\ln n$ tokens when $\mu\le\ln n$.
\end{corollary}

\begin{proof}
    Any given line receives on average $\mu=S/M$ tokens.
    We will determine $\delta$ for which this line receives at most $(1+\delta)\mu$ tokens whp. By the Chernoff inequality
    \[\Pr(X>(1+\delta)\mu))<e^{-\frac{\delta^2\mu}{2+\delta}}. \]
    Case 1: $\mu>\ln n$. Then for $\delta=2\eta,\eta>1$ we have
    \[e^{-\frac{\delta^2\mu}{2+\delta}}<e^{-\frac{4\eta^2\ln n}{2+2\eta}}<e^{-\eta\ln n}=n^{-\eta}.\]
    Case 2: $\mu\le\ln n$. Then for $\delta=\frac{2\ln n\eta}{\mu},\eta>1$ we have
    \[e^{-\frac{\delta^2\mu}{2+\delta}}<e^{-\frac{4\ln^2n\cdot \eta^2}{2\mu+2\ln n\cdot \eta}}<e^{-\eta\ln n}=n^{-\eta}.\]
\end{proof}

\bibliography{references}

%\newpage

%\includegraphics[scale=1]{one-external.jpg}
%\bibliography{bibliography}

%\newpage

\end{document}